\documentclass{article}
\usepackage{spconf,amsmath,amssymb,graphicx,epsfig, cite,algorithm,algorithmic,epstopdf, url}
\usepackage[utf8]{inputenc}
\usepackage[mathscr]{euscript}
\usepackage{color}
\usepackage{bm}
\usepackage{amsthm}
\usepackage[keeplastbox]{flushend}

\def\minwrt[#1]{\underset{#1}{\mathrm{minimize }}}
\def\argminwrt[#1]{\underset{#1}{\text{arg min }}}
\def\argmaxwrt[#1]{\underset{#1}{\text{arg max }}}
\def\maxwrt[#1]{\underset{#1}{\text{maximize }}}
\def\maxemphwrt[#1]{\underset{#1}{\text{\emph{maximize} }}}

\def\mminwrt[#1]{\underset{#1}{\mathrm{min }}}

\newtheorem{theorem}{Theorem}
\newtheorem{remark}{Remark}

\newtheorem{corollary}{Corollary}

\newtheorem{question}{Question}
\newtheorem{conjecture}{Conjecture}

\newcommand{\norm}[1]{\left\lVert#1\right\rVert}
\newcommand{\normtv}[1]{\left\lVert#1\right\rVert_{\mathrm{TV}}}
\newcommand{\abs}[1]{\left|#1\right|}

\def\RC{{\mathbb{C}}}
\def\RR{{\mathbb{R}}}
\def\RZ{{\mathbb{Z}}}

\newcommand{\integerset}[1]{\RZ_{#1}}

\newcommand{\expop}{\mathbb{E}}
\newcommand{\expect}[1]{\expop\left(#1\right)}
\newcommand{\freq}{\theta}

\newcommand{\measureset}[1]{\mathcal{M}_+\left(#1\right)}

\newcommand{\complexmeasureset}[1]{\mathbb{C}\mathcal{M}\left(#1\right)}

\newcommand{\cpol}{\Lambda}

\newcommand{\contfuncs}{\mathcal{C}}

\newcommand{\freqband}{\mathcal{I}_B}

\newcommand{\uncertaintyval}{p^{\star}}

        \makeatletter
        \def\fps@eqnfloat{!t}
        \def\ftype@eqnfloat{4}
        
        \newenvironment{eqnfloat*}
               {\@dblfloat{eqnfloat}}
               {\end@dblfloat}
        \makeatother
\renewenvironment{thebibliography}[1]{%
 \begin{oldthebibliography}{#1}%
 \setlength{\parskip}{0ex}%
 \setlength{\itemsep}{0ex}%
}%
{%
\end{oldthebibliography}%
}%
\begin{document}
\title{Quantifying and Computing Covariance Uncertainty}
\name{Filip Elvander$^*$, Johan Karlsson$^{\dagger}$, Toon van Waterschoot$^*$\thanks{This work was supported in part by the Research Foundation -- Flanders (FWO) grant 12ZD622N, the Swedish Research Council grant 2020-03454, as well as by the European Research Council under the European Union's Horizon 2020 research and innovation program / ERC Consolidator Grant: SONORA (no. 773268). This paper reflects only the authors' views and the Union is not liable for any use that may be made of the contained information.}}
\address{
$^*$Dept. of Electrical Engineering (ESAT-STADIUS), KU Leuven, Belgium\\
$^\dagger$Dept. of Mathematics, KTH Royal Institute of Technology, Sweden\\
emails: firstname.lastname@esat.kuleuven.be,  firstname.lastname@math.kth.se\vspace{-4pt}}

\maketitle
\begin{abstract}
In this work, we consider the problem of bounding the values of a covariance function corresponding to a continuous-time stationary stochastic process or signal. Specifically, for two signals whose covariance functions agree on a finite discrete set of time-lags, we consider the maximal possible discrepancy of the covariance functions for real-valued time-lags outside this discrete grid. Computing this uncertainty corresponds to solving an infinite dimensional non-convex problem. However, we herein prove that the maximal objective value may be bounded from above by a finite dimensional convex optimization problem, allowing for efficient computation by standard methods. Furthermore, we empirically observe that for the case of signals whose spectra are supported on an interval, this upper bound is sharp, i.e., provides an exact quantification of the covariance uncertainty.
\end{abstract}
\vspace{2mm}
\section{Introduction}
Modeling and estimation of the covariance function of wide-sense stationary signals forms an intrinsic and fundamental component of many signal processing algorithms and applications. For example, estimates of temporal and spatial covariance functions are used in radar, sonar, and audio signal processing \cite{KrimV96} for localization and tracking \cite{ElvanderJK18_66,ElvanderHJK20_171} and for performing noise reduction \cite{GannotVMGO17_25,ElvanderAJW19_eusipco}. For the case of temporally narrow-band signals, it is often exploited that time-delays may be represented by wave-form phase-shifts or, equivalently, unit-modulus scaling of the covariance function, as used in, e.g., the Capon method \cite{Capon69} or subspace methods such as MUSIC and ESPRIT \cite{Schmidt79,RoyPK86_34}. For broad-band signals, recent analogous spatial spectral estimators have been proposed that rely on so-called polynomial eigenvalue decompositions \cite{WeissPP18_66,WeissBACPC15_eusipco}, in addition to more classical approaches such as decomposition of the signal into narrow-band components using filtering \cite{Bohme86_10} or beamforming methods such as the steered response power estimator \cite{BrandsteinS97_icassp,DietzenSW21_arxiv}. However, for broad-band signals, the issue of time-delays not being integer multiples of the sampling frequency in general array processing scenarios becomes apparent, not least when generating simulations \cite{ElvanderK21_arxiv}. This then requires rounding or truncation of time-delays \cite{NunesMLBCGSL14_62} or using interpolation, e.g., by means of fractional delay filters which in practice can only be approximate \cite{LaaksoVKL96_13,ValimakiL00_icassp}. Correct interpolation may then have considerable impact on the success of the signal processing task, as the correlation structure is directly related to, e.g., the spatial locations of signal sources \cite{AlmrahWL01_eusipco,DietzenSW21_arxiv,RosseelW21_i3da}. For this reason, one may ask to which extent the covariance function of a signal is determined by its samples at a finite set of discrete lags. It may here be noted that in array processing, the set of lags will always be finite due to the finite number of sensors, i.e., the regime of an infinite sequence of samples as considered in Shannon-Nyquist sampling theorems does not apply.
The related problem of gauging uncertainty in spectral estimation, and in particular the problem of choosing an appropriate metric for the uncertainty, has been considered in \cite{KarlssonG13_58}.
In this work, we consider the problem of quantifying the uncertainty of a continuous-time covariance function given observations of it at a finite set of discrete time-lags. In particular, we study the maximal possible discrepancy in the second-order statistics of any two bandlimited signals whose covariance functions agree on this discrete grid. This is formulated as a worst-case problem where only the bandwidth of the signals as well as the total power are assumed to be known. Although computing this uncertainty or maximal discrepancy corresponds to solving a non-convex optimization problem on an infinite dimensional function space, we show that an upper bound can be constructed by means of a finite-dimensional convex program. We characterize the solution of this problem, as well as its dual, and furthermore conjecture that for the interesting case of the signal band being an interval, the upper bound is actually sharp, i.e., exact. The findings are demonstrated in numerical examples, empirically supporting the conjecture.
\newpage
\section{Covariance uncertainty}
Consider a wide-sense stationary, complex circularly symmetric and zero-mean stochastic process $x$ on the real line. The covariance function $r_x: \RR \to \RC$ is given by
\begin{align*}
	r_x(\tau) \triangleq \expect{x(t)\overline{x(t-\tau)}} =  \int_\RR e^{i2\pi \freq \tau} d\mu_x(\freq),
\end{align*}
for $\tau \in \RR$, where $\mu_x$ is the power spectrum of $x$, and where $i \triangleq \sqrt{-1}$ is the imaginary unit. Herein, we will assume that the spectrum of $x$ is supported on  $\freqband$, which is a union of compact intervals, i.e.,
\begin{align*}
	\int_{\freq \notin \freqband} d\mu_x(\freq) = 0.
\end{align*}
It may here be noted that $\mu_x \in \measureset{\freqband}$, i.e., an element of the set of non-negative measures\footnote{That is, $\mu_x$ may be a generalized integrable function containing, e.g., Dirac deltas.} on $\freqband$. Assume that one has access to the values (or estimates thereof) of the covariance function $r_x$ at a finite, discrete set of lags $\tau \in \integerset{n} \triangleq \left\{-n,-n+1,\ldots,n-1,n  \right\}$ for some integer $n$. Then, one may consider the following question.
\begin{question}\label{qe:cov_uncertainty}
How much can $r_x(\tau)$ for $\tau \in \RR \setminus \integerset{n}$ vary given the spectral support $\freqband$?
\end{question}
Specifically, the question concerns the possible dissimilarity of the second-order statistics of any two processes, or signals, whose covariance functions agree on a given finite set of time-lags. Question~\ref{qe:cov_uncertainty} may be answered by the following optimization problem:
\begin{align}
	\maxwrt[\mu, \nu \in \measureset{\freqband}]& \abs{\int_{\freqband} e^{i2\pi \freq \tau} (d\mu(\freq) - d\nu(\theta))} \label{eq:unspecified_moment_problem} \\
	\text{subject to }& \int_{\freqband} e^{i2\pi \freq k} \left(d\mu(\freq) -d\nu(\freq)\right)= 0 \;,\;\forall k \in \integerset{n}, \notag\\
	& \int_{\freqband} \left(d\mu(\freq) +d\nu(\freq)\right) = 2\sigma^2. \notag
\end{align}
Here, the objective function is the absolute difference at lag $\tau$ between two covariance functions $r_\mu(\tau) \!=\! \int_{\freqband}\! e^{i2\pi \freq \tau} d\mu(\freq)$ and $r_\nu(\tau) = \int_{\freqband} e^{i2\pi \freq \tau} d\nu(\freq)$, with the constraints ensuring that $r_\mu$ and $r_\nu$ agree on $\integerset{n}$. The final constraint ensures that the problem in \eqref{eq:unspecified_moment_problem} is bounded, as the total power of both signals is constrained to be $\sigma^2$. It may here be noted that the values of $r_\mu(k)$ and $r_\nu(k)$ for $k \in \integerset{n}\setminus 0$ are not specified; it is only required that they are equal. Thus, the problem in \eqref{eq:unspecified_moment_problem} corresponds to a worst-case scenario that can be seen as a maximum over all possible covariance functions. The total power $\sigma^2$ then serves as a simple scaling of the problem. Thus, \eqref{eq:unspecified_moment_problem} models the inherent uncertainty in a measurement setup before any measurements are made: the only data in the problem is the expected signal band $\freqband$ and the total power~$\sigma^2$. Being able to compute \eqref{eq:unspecified_moment_problem} would then allow for identifying limitations in, e.g., broad-band array processing. Specifically, as a particular array geometry gives rise to a certain set of (real-valued) time-delays from source to receiver, \eqref{eq:unspecified_moment_problem} quantifies the uncertainty induced by discrete spatio-temporal sampling. This information may then be used as to, e.g., modify the array geometry or determine where to optimally place additional sensors.

It may be noted that \eqref{eq:unspecified_moment_problem} is an infinite-dimensional problem, as it considers optimization on the cone $\measureset{\freqband}$. Furthermore, the problem is non-convex due to the maximization of a convex objective, preventing straight-forward finite-dimensional approximation by gridding. However, the maximal objective of \eqref{eq:unspecified_moment_problem} may be upper-bounded by means of a convex program, as described next.
\section{A computable upper bound}
In order to compute an upper bound to \eqref{eq:unspecified_moment_problem}, define the shorthand $g_\tau(\freq) \triangleq e^{i2\pi\freq\tau}$. Furthermore, let $\contfuncs(\freqband)$ be the set of complex-valued continuous functions on $\freqband$ equipped with the norm $\norm{c} = \sup_{\freq\in\freqband} \abs{c(\freq)}$ for $c\in \contfuncs(\freqband)$. Furthermore, define the subspace $\cpol_n \subset \contfuncs(\freqband)$ as
\begin{align*}
	\cpol_n \triangleq \left\{ Q \mid Q(\freq) = \sum_{k=-n}^n \lambda_k e^{i2\pi \freq k} \;,\; \lambda_k \in \RC \right\},
\end{align*}
i.e., the set of complex trigonometric polynomials of degree at most $n$. Then, the following theorem holds.
\begin{theorem} \label{thm:main_thm}
The maximal objective value of \eqref{eq:unspecified_moment_problem} is upper-bounded by
\begin{align} \label{eq:approx_problem}
	\mminwrt[Q \in \cpol_n] 2\sigma^2\norm{g_\tau - Q}.
\end{align}
\end{theorem}
Here, it may be noted that in contrast to \eqref{eq:unspecified_moment_problem}, the approximation problem in \eqref{eq:approx_problem} is both convex and finite-dimensional due to the finite dimension of the subspace $\cpol_n$. To prove Theorem~\ref{thm:main_thm}, we will use the duality relation between $\contfuncs(\freqband)$ and $\complexmeasureset{\freqband}$, i.e., the set of complex-valued measures on $\freqband$ equipped with the total variation norm
\begin{align*}
	\normtv{\psi} = \int_{\freqband} \abs{d\psi(\freq)}.
\end{align*}
\begin{proof}
By the duality relation, it holds that \cite{Luenberger69}
\begin{align} \label{eq:duality_relation}
	\mminwrt[Q\in \cpol_n] \norm{g_\tau - Q} = \sup_{\psi \in B\cpol_n^\perp} \text{Re}\left( \langle \psi,g_\tau \rangle \right),
\end{align}
where $\langle \psi,g_\tau \rangle = \int_{\freqband} g_\tau(\freq)d\psi(\freq)$, and where $B\cpol_n^\perp$ is the intersection of the unit ball $\left\{ \psi \in \complexmeasureset{\freqband} \mid \normtv{\psi} \leq 1 \right\}$ and the annihilator $\cpol_n^\perp$ of $\cpol_n$, i.e.,
\begin{align*}
	\cpol_n^\perp &\triangleq \left\{ \psi \in \complexmeasureset{\freqband} \mid \langle \psi, Q \rangle = 0 \;\; \forall Q \in \cpol_n  \right\} \\
	&=  \left\{ \psi \in \complexmeasureset{\freqband} \mid \int_{\freqband} e^{i2\pi\freq k}d\psi(\freq) = 0 \;\; \forall k\in \integerset{n}  \right\},
\end{align*}
where the second equality follows from that $\cpol_n$ is a finite-dimensional subspace. Thus, the right-hand side of \eqref{eq:duality_relation} can be written as
\begin{equation} \label{eq:complex_moment_problem}
\begin{aligned}
	\sup_{\psi \in \complexmeasureset{\freqband}}& \mathrm{Re}\left( \int_{\freqband} g_\tau(\freq)d\psi(\freq)  \right) \\
	\text{subject to }& \int_{\freqband} e^{i2\pi \freq k} d\psi(\freq)= 0 \;,\;\forall k \in \integerset{n}, \\
	& \int_{\freqband} \abs{d\psi(\freq)} \leq 1.
\end{aligned}
\end{equation}
Clearly, if $(\mu_0,\nu_0)$ is a solution to \eqref{eq:unspecified_moment_problem} with objective value $\uncertaintyval$,  then
\begin{align*}
	\psi = \frac{e^{-i\varphi}}{2\sigma^2} (\mu_0-\nu_0),
\end{align*}
with $\varphi = \mathrm{arg}\left( \int_{\freqband} e^{i2\pi \freq \tau} (d\mu_0(\freq) - d\nu_0(\theta)) \right)$, is a feasible point of \eqref{eq:complex_moment_problem} with objective value $\uncertaintyval/2\sigma^2$, proving that \eqref{eq:approx_problem} indeed provides an upper bound for \eqref{eq:unspecified_moment_problem}.
\end{proof}
Thus, Theorem~\ref{thm:main_thm} provides a way of computing an upper bound to the covariance uncertainty problem in \eqref{eq:unspecified_moment_problem} by means of a convex optimization program.
Furthermore, for the case when $\freqband$ is symmetric around zero we may characterize an optimal primal-dual pair $(Q_0,\psi_0)$ solving \eqref{eq:approx_problem} and \eqref{eq:complex_moment_problem} according to the following corollary. Here, we define for functions $h$ defined on $\freqband$ the reflection and conjugation operation $h \mapsto h^*$ as $h^*(\freq) = \overline{h(-\freq)}$.
\begin{corollary}\label{cor:alignment}
Let $\freqband$ be symmetric around zero. Then, the optimal $Q_0$ has real coefficients. Furthermore, the optimal $\psi_0$ satisfies $\psi_0 = \psi_0^*$ and can be written as $\psi_0 = \hat{\psi}_0 + i\breve{\psi}_0$, where $\hat{\psi}_0$ and $\breve{\psi}_0$ are real-valued (signed) measures satisfying
\begin{align*}
	\hat{\psi}_0(\freq) = \hat{\psi}_0(-\freq) \;,\; \breve{\psi}_0(\freq) = -\breve{\psi}_0(-\freq).
\end{align*}
The optimal $\psi_0$ is aligned with $g_\tau - Q_0$ and is supported on a subset of
\begin{align*}
	\Omega \triangleq \left\{ \freq \mid \abs{g_\tau(\freq) - Q_0(\freq)} = \norm{g_\tau - Q_0}  \right\},
\end{align*}
which is a point-set symmetric around $\freq = 0$.
\end{corollary}
\begin{proof}
Consider any $Q \in \cpol_n$. Then, $\tilde{Q} = ( Q + Q^* )/2 \in \cpol_n$, and as $g_\tau^* = g_\tau$,
\begin{align*}
	\norm{g_\tau - \tilde{Q}} &= \norm{\frac{1}{2}(g_\tau -Q) +  \frac{1}{2}(g_\tau^* -Q^*)} \\
	&\leq \frac{1}{2}\norm{g_\tau -Q} + \frac{1}{2}\norm{g^*_\tau -Q^*}\\
	&= \norm{g_\tau - Q}
\end{align*}
as \mbox{$\abs{g^*_\tau(\freq)\!-\!Q^*(\freq)} \!=\! \abs{\overline{g_\tau(\!-\freq)} \!-\!\overline{Q(\!-\freq)}} = \abs{g_\tau(\!-\freq) \!-\!Q(\!-\freq)}$}. Thus, any candidate solution $Q$ can always be improved to a solution $\tilde{Q}$ satisfying $\tilde{Q} = \tilde{Q}^*$. As
\begin{align*}
	\tilde{Q}(\freq) = \sum_{k=-n}^n \lambda_k e^{i2\pi \freq k } \;,\; \tilde{Q}^*(\freq) = \sum_{k=-n}^n \overline{\lambda_k} e^{i2\pi \freq k },
\end{align*}
this implies $\lambda_k = \overline{\lambda_k} \in \RR$.

The alignment follows directly from \eqref{eq:duality_relation}, which for the case of continuous functions and complex measures implies that $\psi_0$ is only supported where $\abs{g_\tau - Q_0}$ is maximal. As $g_\tau - Q_0$ is a linear combination of finitely many sinusoids, it follows that $\abs{g_\tau - Q_0}$ can only be maximal on an interval if it is identically equal to $\norm{g_\tau - Q_0}$. Thus, the maximizing frequencies $\Omega$ constitute a set of isolated points. The symmetry of $\Omega$ follows from the fact that $Q_0$ has real coefficients, implying $g_\tau-Q_0 = (g_\tau - Q_0)^*$. To show $\psi_0 = \psi_0^*$, consider any feasible $\psi$ and construct $\tilde{\psi} = (\psi + \psi^*)/2$. Defining the functional $f: \complexmeasureset{\freqband} \to \RR$ as 
\begin{align*}
	f(\psi) = \mathrm{Re}\left( \int_{\freqband} g_\tau(\freq)d\psi(\freq)  \right),
\end{align*}
it is readily verified that $f(\tilde{\psi}) = f(\psi)$ due to the symmetric integration set and that $g_\tau = g_\tau^*$. Furthermore, $\tilde{\psi}$ clearly satisfies the linear constraints and by the convexity of the total variation norm, $||\tilde{\psi}||_{\mathrm{TV}} \leq \normtv{\psi}$. By the linearity of $f$, $\tilde{\psi}$ may then be scaled as to obtain a feasible solution with improved objective. The decomposition into a symmetric real and antisymmetric imaginary part follows directly.
\end{proof}
\begin{remark}
It may be noted that any problem where $\freqband$ is symmetric around a center frequency $\freq_c$ can be mapped to an equivalent problem on the form considered in Corollary~\ref{cor:alignment}.
To see this, note that shifting the frequency axis by $\freq_c$ corresponds to a constant phase-shift of the objective and constraints of \eqref{eq:unspecified_moment_problem}, thus not affecting neither objective value nor feasibility. Thus, a solution $(\mu,\nu)$ for the symmetric problem corresponds to a solution with center frequency $\freq_c$ by a shift $(\mu(\cdot+\freq_c),\nu(\cdot+\freq_c))$.
\end{remark}
Although Theorem~\ref{thm:main_thm} provides an upper bound on the covariance uncertainty we have empirically observed a stronger result for a special case of the symmetric sets considered in Corollary~\ref{cor:alignment}: when $\freqband$ is an interval, the bound appears to be sharp. We state this observation in the following conjecture.
\begin{conjecture} \label{conj:equality}
Let $\freqband$ be an interval. Then, the objective values of \eqref{eq:unspecified_moment_problem} and \eqref{eq:approx_problem} coincide.
\end{conjecture}
As an alternative to Theorem~\ref{thm:main_thm} and Conjecture~\ref{conj:equality}, it may be noted that \eqref{eq:unspecified_moment_problem} may be reformulated as the non-convex problem
\begin{align}
	\maxwrt[\mu, \nu \in \measureset{\freqband}, \phi\in \mathbb{T}]& \text{Re } \left(\phi\int_{\freqband} e^{i2\pi \freq \tau} (d\mu(\freq) - d\nu(\theta))\right)\label{eq:exact_bound_with_tideous_computation} \\
	\text{subject to }& \int_{\freqband} \!e^{i2\pi \freq k} \left(d\mu(\freq) -d\nu(\freq)\right)= 0 \;,\forall k \in \integerset{n}, \notag\\
	& \int_{\freqband} \left(d\mu(\freq) +d\nu(\freq)\right) = 2\sigma^2, \notag
\end{align}
where $\mathbb{T} \triangleq \{ \phi \in \RC \mid \abs{\phi} = 1 \}$. However, this problem is convex if $\phi$ is kept fixed, and thus the optimum of \eqref{eq:unspecified_moment_problem} may be obtained by solving the restricted convex problem for each $\phi \in \mathbb{T}$. In practice, an approximate solution can be computed for each $\phi$ in a fine grid on $\mathbb{T}$, although it may be noted that this is a tedious and computationally heavy method for computing the covariance uncertainty.
\section{Numerical illustrations}
%
\begin{figure}[t!]
        \centering
        \vspace{1mm}
            \includegraphics[width=0.45\textwidth]{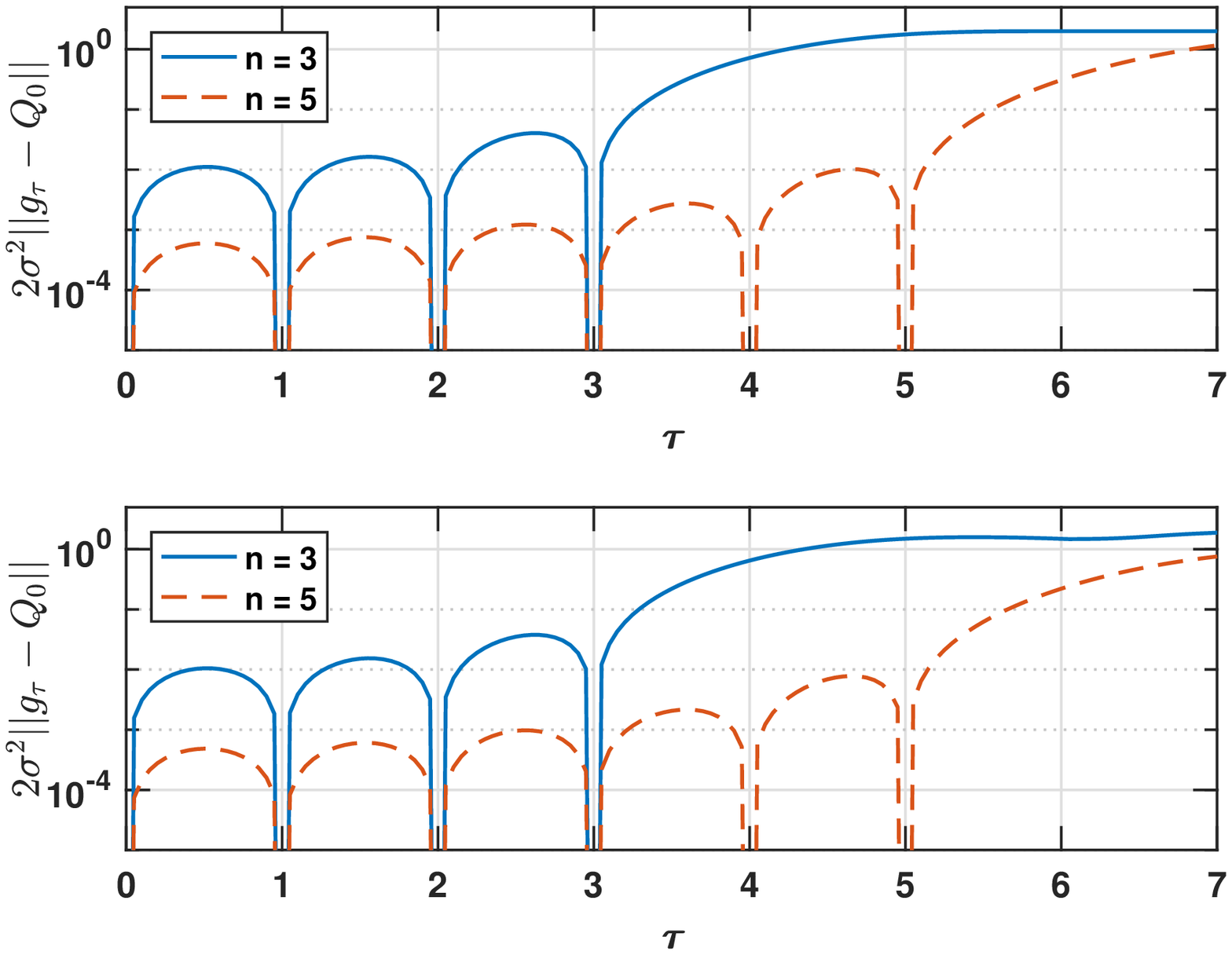}
            \vspace{.75mm}
           \caption{Covariance uncertainty bound $2\sigma^2\norm{g_\tau - Q_0}$ as function of $\tau$ for $n = 3$ and $n = 5$. Top panel: $\freqband = [-0.3,0.3]$. Bottom panel: $\freqband = [-0.3,-0.1] \cup [0.05,0.3]$.}
            \label{fig:cov_uncertainty}
\end{figure}
%
In this section, we provide some illustrations of the upper bound on covariance uncertainty as given by Theorem~\ref{thm:main_thm}. In particular, we show both cases where the bound is not tight and cases supporting the statement of Conjecture~\ref{conj:equality}. The practical computations are performed by discretizing the frequency axis on the set $\freqband$ and then solving \eqref{eq:approx_problem} using the general-purpose convex optimization package CVX \cite{cvx_new}. To check if the bound is sharp, we approximate the exact uncertainty \eqref{eq:unspecified_moment_problem} by solving \eqref{eq:exact_bound_with_tideous_computation} for a fine grid on $\mathbb{T}$.

We here consider two different scenarios; in the first $\freqband = [-0.3,0.3]$, and in the second $\freqband = [-0.3,-0.1] \cup [0.05,0.3]$. It may here be noted that the first scenario conforms with the conditions of Conjecture~\ref{conj:equality}, and we therefore expect the bound to be sharp, whereas this is not the case for the second scenario. Without loss of generality, we fix $\sigma^2 = 1$.

Figure~\ref{fig:cov_uncertainty} displays the bound \eqref{eq:approx_problem} for $\tau \in [0,7]$ for the two scenarios for $n = 3$ and $n = 5$. As can be seen, as $\tau$ increases beyond $n$, the bound approaches the trivial bound $2\sigma^2 = 2$. It may here be noted that the bound does not necessarily have local maxima located exactly in the middle between two specified covariances. For example, for $n = 5$, the maximal bound for $\tau \in [4,5]$ is not at $\tau = 4.5$ but slightly higher.
Although the bounds for the scenarios of interval and non-interval $\freqband$ behave qualitatively similar, a difference appears when considering the gap between the bound in \eqref{eq:approx_problem} and the approximation in \eqref{eq:exact_bound_with_tideous_computation} when solved for a fine grid of $\phi \in \mathbb{T}$, i.e.,
\begin{align*}
	2\sigma^2\norm{g_\tau - Q_0} - \text{Re } \left(\phi_0\int_{\freqband} e^{i2\pi \freq \tau} (d\mu_0(\freq) - d\nu_0(\theta))\right),
\end{align*}
where $(\mu_0,\nu_0)$ is an optimal pair for \eqref{eq:exact_bound_with_tideous_computation} for the maximizing $\phi_0 \in \mathbb{T}$. This gap is displayed in Figure~\ref{fig:gap} for the same scenarios as in Figure~\ref{fig:cov_uncertainty} . As can be seen, for the case of $\freqband$ being an interval, the empirical gap is erratic and small enough to be attributed to the tolerance of the numerical solver. In contrast, for the case of an asymmetric $\freqband$, the empirical gap becomes relatively large when the lag $\tau$ increases beyond the largest specified lag, as well as appears to be fairly smooth as a function of $\tau$. It may however be noted that for $\tau \leq n$, the gap between the bound and the exact uncertainty appears small.
Taken together, this gives empirical support for Conjecture~\ref{conj:equality}, i.e., the bound \eqref{eq:approx_problem} is sharp when $\freqband$ is an interval.
%
\begin{figure}[t!]
        \centering
        \vspace{1mm}
            \includegraphics[width=0.45\textwidth]{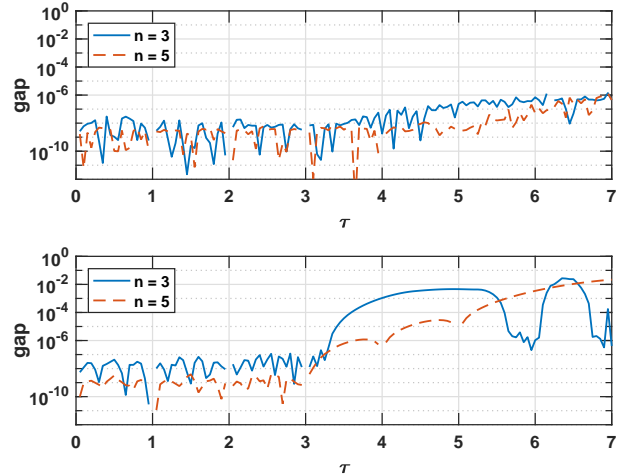}
            \vspace{.75mm}
           \caption{The gap between the bound \eqref{eq:approx_problem} and the maximal value of \eqref{eq:exact_bound_with_tideous_computation} when solved on a fine grid of $\phi \in \mathbb{T}$ for $n = 3$ and $n = 5$. Top panel: $\freqband = [-0.3,0.3]$. Bottom panel: $\freqband = [-0.3,-0.1] \cup [0.05,0.3]$.}
            \label{fig:gap}
\end{figure}

%
%
\section{Conclusions and future work}
In this work, we have shown that the maximal discrepancy between any two covariance functions corresponding to signals of a certain bandwidth may be bounded from above by a finite-dimensional convex program. Furthermore, we have empirically demonstrated that for the case of signal bands that are intervals, the bound appears to be sharp. Proving this rigorously, as well as extending the results to general spatio-temporal covariance functions, is planned for the future.
\newpage
\bibliographystyle{IEEEbib}
\bibliography{IEEEabrv,ElvanderKW22_icassp_arxiv.bbl}
%
\end{document}